\documentclass[runningheads,envcountsame,oribibl,orivec]{llncs}
\usepackage{amsmath,amssymb}
\usepackage{graphicx,float,url,booktabs}
\usepackage[usenames]{xcolor}
\usepackage[english]{babel}
\usepackage{tabularx,xspace}
\usepackage{algo}
\usepackage[caption=false]{subfig}
\usepackage{cite}

\graphicspath{{figures/}}

\definecolor{sepia}{hsb}{0.05,0.75,0.65}
\everymath{\color{sepia}}

\title{
 Computing the Greedy Spanner in Linear Space
}

\author{Sander P. A. Alewijnse \and
        Quirijn W. Bouts \and \\
        Alex P. ten Brink \and
        Kevin Buchin
}

\institute{
Eindhoven University of Technology, The Netherlands,
\path|{s.p.a.alewijnse,	q.w.bouts,a.p.t.brink}@student.tue.nl|,
\path|k.a.buchin@tue.nl|
}

\authorrunning{S.P.A.~Alewijnse, Q.W.~Bouts, A.P.~ten Brink, K.~Buchin}

\expandafter\ifx\csname pdfpagebox\endcsname\relax\else
\fi

\spnewtheorem{observation}[theorem]{Observation}{\bfseries}{\itshape}
\spnewtheorem{fact}[theorem]{Fact}{\bfseries}{\itshape}

{\unskip\nobreak\hskip 1em plus 1fil\nobreak$\Box$
\parfillskip=0pt%
\endtrivlist}
{\unskip\nobreak\hskip 2em plus 1fil\nobreak$\Box$
\parfillskip=0pt%
\endtrivlist}

\let\doendproof\endproof
\renewcommand\endproof{~\hfill$\qed$\doendproof}

\newcommand{\Reals}{\mathbb{R}}

\begin{document}

\mainmatter

\maketitle

\begin{abstract}
The greedy spanner is a high-quality spanner: its total weight, edge count and maximal degree are asymptotically optimal and in practice significantly better than for any other spanner with reasonable construction time. Unfortunately, all known algorithms that compute the greedy spanner of $n$ points use $\Omega(n^2)$ space, which is impractical on large instances. To the best of our knowledge, the largest instance for which the greedy spanner
was computed so far has about 13,000 vertices.

We present a $O(n)$-space algorithm that computes the same spanner for points in $\mathbb{R}^d$ running in $O(n^2 \log^2 n)$ time for any fixed stretch factor and dimension. We discuss and evaluate a number of optimizations to its running time, which allowed us to compute the greedy spanner on a graph with a million vertices. To our knowledge, this is also the first algorithm for the greedy spanner with a near-quadratic running time guarantee that has actually been implemented.
\end{abstract}

\section{Introduction}

A $t$-spanner on a set of points, usually in the Euclidean plane, is a graph on these points that is a `$t$-approximation' of the complete graph, in the sense that shortest routes in the graph are at most $t$ times longer than the direct geometric distance. The spanners considered in literature have only $O(n)$ edges as opposed to the $O(n^2)$ edges in the complete graph, or other desirable properties such as bounded diameter or bounded degree, which makes them a lot more pleasant to work with than the complete graph.

Spanners are used in wireless network design~\cite{GaoGHZZ05}: for example, high-degree routing points in such networks tend to have problems with interference, so using a spanner with bounded degree as network avoids these problems while maintaining connectivity. They are also used as a component in various other geometric algorithms, and are used in distributed algorithms. Spanners were introduced in network design \cite{JGT:JGT3190130114} and geometry \cite{Chew1989}, and have since been subject to a considerable amount of research~\cite{DilationandDetours,Narasimhan:2007:GSN:1208237}.

There exists a large number of constructions of $t$-spanners that can be parameterized with arbitrary $t>1$. They have different strengths and weaknesses: some are fast to construct but of low quality ($\Theta$-graph, which has no guarantees on its total weight), others are slow to construct but of high quality (greedy spanner, which has low total weight and maximum degree), some have an extremely low diameter (various dumbbell based constructions) and some are fast to construct in higher dimensions (well-separated pair decomposition spanners). See for example \cite{Narasimhan:2007:GSN:1208237} for detailed expositions of these spanners and their properties.

\begin{figure}[t]\centering
\includegraphics{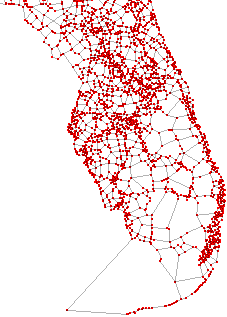}
\includegraphics{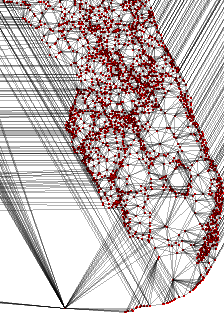}
\caption{The left rendering shows the greedy spanner on the USA, zoomed in on Florida, with $t=2$. The dataset has 115,475 vertices, so it was infeasible to compute this graph before. The right rendering shows the $\Theta$-graph on the USA, zoomed in on Florida, with $k=6$ for which it was recently proven it achieves a dilation of 2.}
\label{figure:greedy}
\end{figure}

The greedy spanner is one of the first spanner algorithms that was considered, and it has been subject to a considerable amount of research regarding its properties and more recently also regarding computing it efficiently. This line of research resulted in a $O(n^2 \log n)$ algorithm \cite{BoseCFMS2010} for metric spaces of bounded doubling dimension (and therefore also for Euclidean spaces). There is also an algorithm with $O(n^3 \log n)$ worst-case running time that works well in practice~\cite{FarshiG09}. Its running time tends to be near-quadratic in practical cases, but there are examples on which its running time is $\Theta(n^3 \log n)$. Its space usage is $\Theta(n^2)$.

Among the many spanner algorithms known, the greedy spanner is of special interest because of its exceptional quality: its size, weight and degree are asymptotically optimal, and also in practice are better than any other spanner construction algorithms with reasonable running times. For example, it produces spanners with about ten times as few edges, twenty times smaller total weight and six times smaller maximum degree as its closest well-known competitor, the $\Theta$-graph, on uniform point sets. The contrast is clear in Fig.~\ref{figure:greedy}. Therefore, a method of computing it more efficiently is of considerable interest.

We present an algorithm whose space usage is $\Theta(n)$ whereas existing algorithms use $\Theta(n^2)$ space, while being only a logarithmic factor slower than the fastest known algorithm, thus answering a question left open in \cite{BoseCFMS2010}.
Our algorithm makes the greedy spanner practical to compute for much larger inputs than before: this used to be infeasible on graphs of over 15,000 vertices. In contrast, we tested our algorithm on instances of up to 1,000,000 points, for which previous algorithms would require multiple terabytes of memory. Furthermore, with the help of several optimizations we will present, the algorithm is also fast in practice, as our experiments show.

The method used to achieve this consists of two parts: a framework that uses linear space and near-linear time, and a subroutine using linear space and near-linear time, which is called a near-linear number of times by the framework. The subroutine solves the \emph{bichromatic closest pair with dilation larger than $t$} problem. If there is an algorithm with a sublinear running time for this subproblem (possibly tailored to our specific scenario), then our framework immediately gives an asymptotically faster algorithm than is currently known. This situation is reminiscent to that of the minimum spanning tree, for which it is known that it is essentially equivalent to the \emph{bichromatic closest pair} problem.

The rest of the paper is organized as follows. In Section \ref{section:preliminaries} we review a number of well-known definitions, algorithms and results. In Section \ref{section:properties} we give the properties of the WSPD and the greedy spanner on which our algorithm is based.
In Section \ref{section:algorithm} we present our algorithm and analyse its running time and space usage. In Section \ref{section:practical} we discuss our optimizations of the algorithm. Finally, in Section \ref{section:results} we present our experimental results and compare it to other algorithms.

\section{Notation and preliminaries} \label{section:preliminaries}

Let $V$ be a set of points in $\Reals^d$, and let $t \in \Reals$ be the intended dilation ($1 < t$). Let $G = (V, E)$ be a graph on $V$. For two points $u, v \in V$, we denote the Euclidean distance between $u$ and $v$ by $|uv|$, and the distance in $G$ by $\delta_G(u, v)$. If the graph $G$ is clear from the context we will simply write $\delta(u,v)$. The \emph{dilation} of a pair of points is $t$ if $\delta(u, v) \leq t \cdot |uv|$. A graph $G$ has dilation $t$ if $t$ is an upper bound for the dilations of all pairs of points. In this case we say that $G$ is a \emph{$t$-spanner}. To simplify the analysis, we assume without loss of generality that $t<2$.

We will often say that two points $u, v \in V$ \emph{have a $t$-path} if their dilation is $t$. A pair of points \emph{is without $t$-path} if its dilation is not $t$. When we say a pair of points $(u, v)$ is the \emph{closest} or \emph{shortest} pair among some set of points, we mean that $|uv|$ is minimal among this set. We will talk about \emph{a Dijkstra computation} from a point $v$ by which we mean a single execution of the single-source shortest path algorithm known as Dijkstra's algorithm from $v$.

Consider the following algorithm that was introduced by Keil \cite{Keil:1988:ACE:61764.61787}:

\pagebreak

\begin{algorithm}{GreedySpannerOriginal}[V, t]{\label{algo:greedyorig}}
  $E \gets \emptyset$
  \\ \qfor every pair of distinct points $(u, v)$ in ascending order of $|uv|$
  \\ \qdo \qif $\delta_{(V, E)}(u, v) > t \cdot |uv|$
       \\ \qthen add $(u, v)$ to $E$
          \qendif
     \qendfor
  \\ \qreturn $E$
\end{algorithm}

Obviously, the result of this algorithm is a $t$-spanner for $V$. The resulting graph is called the \emph{greedy spanner} for $V$, for which we shall present a more efficient algorithm than the above.

We will make use of the Well-Separated Pair Decomposition, or WSPD for short, as introduced by Callahan and Kosaraju in \cite{Callahan95dealingwith,Callahan:1995:DMP:200836.200853}. A WSPD is parameterized with a separation constant $s \in \Reals$ with $s > 0$. This decomposition is a set of pairs of nonempty subsets of $V$. Let $m$ be the number of pairs in a decomposition. We can number the pairs, and denote every pair as $\{ A_i, B_i \}$ with $1 \leq i \leq m$. Let $u$ and $v$ be distinct points, then we say that $(u, v)$ is `in' a well-separated pair $\{ A_i, B_i \}$ if $u \in A_i$ and $v \in B_i$ or $v \in A_i$ and $u \in B_i$. A decomposition has the property that for every pair of distinct points $u$ and $v$, there is exactly one $i$ such that $(u,v)$ is in $\{ A_i, B_i \}$.

For two nonempty subsets $X_k$ and $X_l$ of $V$, we define $\min(X_k, X_l)$ to be the shortest distance between the two circles around the bounding boxes of $X_k$ and $X_l$ and $\max(X_k, X_l)$ to be the longest distance between these two circles. Let $diam(X_k)$ be the diameter of the circle around the bounding box of $X_k$. Let $\ell(X_k, X_l)$ be the distance between the centers of these two circles, also named the \emph{length} of this pair. For a given separation constant $s \in \Reals$ with $s>0$ as parameter for the WSPD, we require that all pairs in a WSPD are $s$-well-separated, that is, $\min(A_i, B_i) \geq s \cdot \max(diam(A_i), diam(B_i))$ for all $i$ with $1 \leq i \leq m$.

It is easy to see that $\max(X_k, X_l) \leq \min(X_k, X_l) + diam(X_k) + diam(X_l) \leq (1 + 2/s) \min(X_k, X_l)$. As $t < 2$ and as we will pick $s = \frac{2t}{t-1}$ later on, we have $s > 4$, and hence $\max(X_k, X_l) \leq 3/2 \min(X_k, X_l)$. Similarly, $\ell(X_k, X_l) \leq \min(X_k, X_l) + diam(X_k)/2 + diam(X_l)/2 \leq (1 + 1/s) \min(X_k, X_l)$ and hence $\ell(X_k, X_l) \leq 5/4 \min(X_k, X_l)$.

For any $V$ and any $s>0$, there exists a WSPD of size $m=O(s^d n)$ that can be computed in $O(n \log n + s^d n)$ time and can be represented in $O(s^d n)$ space \cite{Callahan95dealingwith}. Note that the above four values ($\min$, $\max$, $diam$ and $\ell$) can easily be precomputed for all pairs with no additional asymptotic overhead during the WSPD construction.

\section{Properties of the greedy spanner and the WSPD} \label{section:properties}

In this section we will give the idea behind the algorithm and present the properties of the greedy spanner and the WSPD that make it work. We assume we have a set of points $V$ of size $n$, an intended dilation $t$ with $1 < t < 2$ and a WSPD with separation factor $s = \frac{2t}{t-1}$, for which the pairs are numbered $\{A_i, B_i\}$ with $1 \leq i \leq m$, where $m = O(s^d n)$ is the number of pairs in the WSPD.

The idea behind the algorithm is to change the original greedy algorithm to work on well-separated pairs rather than edges. We will end up adding the edges in the same order as the greedy spanner. We maintain a set of 'candidate' edges for every well-separated pair such that the shortest of these candidates is the next edge that needs be added. We then recompute a candidate for some of these well-separated pairs. We use two requirements to decide on which pairs we perform a recomputation, that together ensure that we do not do too many recomputations, but also that we do not fail to update pairs which needed updating.

We now give the properties on which our algorithm is based.

\begin{observation}[{{Bose et al.~\cite[Observation 1]{BoseCFMS2010}}}] \label{observation:wspdgreedy}
 For every $i$ with $1 \leq i \leq m$, the greedy spanner includes at most one edge $(u, v)$ with $(u, v)$ in $\{ A_i, B_i \}$.
\end{observation}
\begin{proof}
This observation is not fully identical to \cite{BoseCFMS2010} as our definition of well-separatedness is slightly different than theirs. However, their proof uses only properties of their Lemma 2, which still holds true using our definitions as proven in \ref{section:preliminaries}. Their Lemma 2 is almost Lemma 9.1.2 in \cite{Narasimhan:2007:GSN:1208237}, whose definitions of well-separatedness is near identical to ours, except that they use radii rather than diameters and hence have different constants.
\end{proof}

An immediate corollary is:

\begin{observation}[{{Bose et al.~\cite[Corollary 1]{BoseCFMS2010}}}] \label{observation:greedyedges}
 The greedy spanner contains at most $O\left(\frac{1}{(t-1)^d} n\right)$ edges.
\end{observation}

\begin{lemma} \label{lemma:dijkstra}
Let $E$ be some edge set for $V$. For every $i$ with $1 \leq i \leq m$, we can compute the closest pair of points $(u, v) \in A_i \times B_i$ among the pairs of points with dilation larger than $t$ in $G=(V, E)$ in $O(\min(|A_i|, |B_i|)(|V|\log|V| + |E|))$ time and $O(|V|)$ space.
\end{lemma}

\begin{proof}
Assume without loss of generality that $|A_i| \leq |B_i|$. We perform a Dijkstra computation for every point $a \in A_i$, maintaining the closest point in $|B_i|$ such that its dilation with respect to $a$ is larger than $t$ over all these computations. To check whether a point that is considered by the Dijkstra computation is in $|B_i|$, we precompute a boolean array of size $|V|$, in which points in $|B_i|$ are marked as true and the rest as false. This costs $O(|V|)$ space, $O(|V|)$ time and achieves a constant lookup time. A Dijkstra computation takes $O(|V|\log|V|+|E|)$ time and $O(|V|)$ space, but this space can be reused between computations.
\end{proof}

\begin{fact}[{{Callahan~\cite[Chapter 4.5]{Callahan95dealingwith}}}]\label{fact:manyone}
$\sum_{i=1}^m \min(|A_i|, |B_i|) = O(s^d n \log n)$
\end{fact}
\begin{proof}
This is not stated explicitly in \cite{Callahan95dealingwith}, but it is a direct consequence of the construction of the many-one realization: it is proven that the many-one realization consists of $O(s^d n \log n)$ pairs, and the construction splits every pair in the canonical realization into $\min(|A_i|, |B_i|)$ pairs, so the lemma follows.
\end{proof}

\begin{observation} \label{observation:faraway}
Let $E$ be some edge set for $V$. Let $(a, b) \in E$. Let $c \in V$ and $d \in V$ be points such that $|ac|, |ad|, |bc|, |bd| > t|cd|$. Then any $t$-path between $c$ and $d$ will not use the edge $(a, b)$.
\end{observation}

\begin{proof}
This directly follows from the fact that $c$ and $d$ are so far away from $a$ and $b$ that just getting to either $a$ or $b$ is already longer than allowed for a $t$-path.
\end{proof}

\begin{fact} \label{fact:packing}
Let $\gamma$ and $\ell$ be positive real numbers, and let $\{ A_i, B_i \}$ be a well-separated pair in the WSPD with length $\ell(A_i, B_i) = \ell$. The number of well-separated pairs $\{ A_i', B_i' \}$ such that the length of the pair is in the interval $[\ell/2,2\ell]$ and at least one of $R(A_i')$ and $R(B_i')$ is within distance $\gamma \ell$ of either $R(A_i)$ or $R(B_i)$ is less than or equal to $c_{s \gamma} = O\left(s^d(1+\gamma s)^d\right)$.
\end{fact}

\begin{proof}
This follows easily from a very similar statement, Lemma 11.3.4 in \cite{Narasimhan:2007:GSN:1208237}. We first note the differences between their definitions and ours. Their statement involves \emph{dumbbells}, but these are really the same as our well-separated pairs. Also, their definition of well-separatedness is different from ours in that they use the radius and we the diameter. We can easily amend this by observing that a WSPD with separation factor $s$ using our definitions is identical to a WSPD with separation factor $2s$ using their definitions. This means our constant $c_{s \gamma}$ is larger than their constant because our $s$ is doubled, but this is asymptotically irrelevant.

Now we shall see how our statement follows from theirs. Their interval is $[\ell, 2\ell]$ and they allow the length of $\{ A_i', B_i' \}$ to vary in the interval $[\ell, 2\ell]$. Let $\ell' = \ell(A_i', B_i')$, then we can obtain our fact by invoking their lemma twice, first by setting $\ell = \ell'$ (resulting in an upper bound on the number such pairs $n_1$) and then by setting $\ell = \ell'/2$ (resulting in an upper bound on the number of such pairs $n_2$). This counts the number of pairs with lengths in the interval $[\ell/2, 2\ell]$. These are exactly the pairs we are interested in, and hence the constant we obtain is $n_1 + n_2$, thus proving the fact.
\end{proof}

This concludes the theoretical foundations of the algorithm. We will now present the algorithm and analyze its running time.

\section{Algorithm} \label{section:algorithm}

We will now describe algorithm \ref{algo:Greedy} in detail. It first computes the WSPD for $V$ with $s = \frac{2t}{t-1}$ and sorts the resulting pairs according to their smallest distance $\min(A_i, B_i)$. It then alternates between calling the \qproc{FillQueue} procedure that attempts to add well-separated pairs to a priority queue $Q$, and removing an element from $Q$ and adding a corresponding edge to $E$. If $Q$ is empty after a call to \qproc{FillQueue}, the algorithm terminates and returns $E$.

We assume we have a procedure \qproc[i]{ClosestPair} that for the $i$th well-separated pair computes the closest pair of points without $t$-path in the graph computed so far, as presented in Lemma \ref{lemma:dijkstra}, and returns this pair, or returns \qnil if no such pair exists. For the priority queue $Q$, we let $\min(Q)$ denote the value of the key of the minimum of $Q$. Recall that $m = O(s^d n)$ denotes the number of well-separated pairs in the WSPD that we compute in the algorithm.

We maintain an index $i$ into the sorted list of well-separated pairs. It points to the smallest untreated well-separated pair -- we treat the pairs in order of $\min(A_i, B_i)$ in the \qproc{FillQueue} procedure. When we treat a pair $\{ A_i, B_i \}$, we call \qproc[i]{ClosestPair} on it, and if it returns a pair $(u, v)$, we add it to $Q$ with key $|uv|$. We link entries in the queue, its corresponding pair $\{ A_i, B_i \}$ and $(u, v)$ together so they can quickly be requested. We stop treating pairs and return from the procedure if we have either treated all pairs, or if $\min(A_i, B_i)$ is larger than the key of the minimal entry in $Q$ (if it exists).

After extracting a pair of points $(u, v)$ from $Q$, we add it to $E$. Then, we update the information in $Q$: for every pair $\{ A_j, B_j \}$ having an entry in $Q$ for which either bounding box is at most $t|uv|$ away from $\{ A_i, B_i \}$, we recompute \qproc[j]{ClosestPair} and updates its entry in $Q$ as follows. If the recomputation returns \qnil, we remove its entry from $Q$. If it returns a pair $(u', v')$, we link the entry of $j$ in $Q$ with this new pair and we increase the key of its entry to $|u'v'|$.

\begin{algorithm}{FillQueue}[Q, i]{}
  \qwhile \label{algo:guard} $i \leq m$, and either $\min(A_i B_i) \leq \min(Q)$ or $Q$ is empty
  \\ \qdo $p \gets $ \qproc[i]{ClosestPair}
       \\ \qif $p$ is not \qnil, but a pair $(u, v)$
       \\ \qthen add $(u, v)$ to $Q$ with key $|uv|$, and associate this entry with $\{ A_i, B_i \}$
          \qendif
       \\ $i \gets i + 1$
  \qend
\end{algorithm}

\begin{algorithm}{GreedySpanner}[V, t]{\label{algo:Greedy}}
  Compute the WSPD $W$ for $V$ with $s = \frac{2t}{t-1}$, and let $\{ A_i, B_i \}$ be the resulting pairs, $1 \leq i \leq m$
  \\ Sort the pairs $\{ A_i, B_i \}$ according to $\min(A_i, B_i)$
  \\ $E \gets \emptyset$
  \\ $Q \gets $ an empty priority queue
  \\ $i \gets 1$
  \\ \label{algo:fill} \qproc[Q, i]{FillQueue}
  \\ \qwhile $Q$ is not empty
  \\ \qdo extract the minimum from Q, let this be $(u, v)$
       \\ \label{algo:add} add $(u, v)$ to $E$
       \\ \qfor \label{algo:loop} all pairs $\{ A_j, B_j \}$ with an entry in $Q$ for which either bounding box is at most $t|uv|$ away from either $u$ or $v$
       \\ \qdo $p \gets $ \qproc[j]{ClosestPair}
            \\ \label{algo:update1} \qif $p$ is \qnil, remove the entry in $Q$ associated with $\{ A_j, B_j \}$ from $Q$ \qendif
            \\ \label{algo:update2} \qif $p$ is a pair $(u', v')$, update the entry in $Q$ associated with $\{ A_j, B_j \}$ to contain $(u', v')$ and increase its key to $|u'v'|$ \qendif
          \qendfor
       \\ \qproc[Q, i]{FillQueue}
     \qend
  \\ \qreturn $E$
\end{algorithm}

\pagebreak

We now prove correctness and a bound on the running time of the algorithm.

\begin{theorem} \label{theorem:correctness}
Algorithm \qproc{GreedySpanner} computes the greedy spanner for dilation $t$.
\end{theorem}

\begin{proof}
We will prove that if the algorithm adds $(u, v)$ to $E$ on line \ref{algo:add}, then $(u, v)$ is the closest pair of points without a $t$-path between them. The greedy spanner consists of exactly these edges and hence this is sufficient to prove the theorem.

It is obvious that if we call \qproc[i]{ClosestPair} on every well-separated pair and take the closest pair of the non-\qnil results, then that would be the closest pair of points without a $t$-path between them. Our algorithm essentially does this, except it does not recalculate \qproc[i]{ClosestPair} for every pair after every added edge, but only for specific pairs. We will prove that the calls it does not make do not change the values in $Q$. Our first optimization is that if a call \qproc[i]{ClosestPair} returns \qnil it will always return \qnil, so we need not call \qproc[i]{ClosestPair} again, which is therefore a valid optimization.

The restriction `for which either bounding box is at most $t|uv|$ away from either $u$ or $v$' from the for-loop on line \ref{algo:loop} is the second optimization. Its validity is a direct consequence of Observation \ref{observation:faraway}: all pairs of points in such well-separated pairs are too far away to use the newly-added edge to gain a $t$-path. Therefore re-running \qproc[i]{ClosestPair} and performing lines \ref{algo:update1} and \ref{algo:update2} will not change any entries in $Q$ as claimed.

As $\min(A_i B_i)$ is a lower bound on the minimal distance between any two points $(a, b)$ in $\{ A_i, B_i \}$, it immediately follows that calling \qproc[Q, i]{FillQueue} on a pair $\{ A_i, B_i \}$ with $\min(A_i B_i) > \min(Q)$ cannot possibly yield a pair that can cause $\min(Q)$ to become smaller. As the pairs are treated in order of $\min(A_i B_i)$, this means the optimization that is the condition on line~\ref{algo:guard} in~\qproc[Q, i]{FillQueue} is a valid optimization. This proves the theorem.
\end{proof}

We will now analyze the running time and space usage of the algorithm. We will use the observations in Section \ref{section:properties} to bound the amount of work done by the algorithm.

\begin{lemma} \label{lemma:closestpaircount}
For any well-separated pair $\{ A_i, B_i \}$ {\normalfont($1 \leq i \leq m$)}, the number of times \qproc[i]{ClosestPair} is called is at most $1 + c_{s t}$.
\end{lemma}

\begin{proof}
\qproc[i]{ClosestPair} is called once for every $i$ in the \qproc{FillQueue} procedure. \qproc[i]{ClosestPair} may also be called after an edge is added to the graph. We will show that if a well-separated pair $\{ A_j, B_j \}$ causes \qproc[i]{ClosestPair} to be called, then $\ell(A_j, B_j) \in [\ell(A_i, B_i)/2, 2\ell(A_i, B_i)]$. Then, by the condition of line \ref{algo:loop}, the collection of pairs that call \qproc[i]{ClosestPair} satisfy the requirements of Fact \ref{fact:packing} by setting $\gamma = t$, so we can conclude this happens only $c_{s t}$ times. The lemma follows.

We will now show that $\ell(A_j, B_j) \in [\ell(A_i, B_i)/2, 2\ell(A_i, B_i)]$. Recall the following from Section \ref{section:preliminaries}:
\begin{align*}
\max(A_i, B_i) &\leq \frac{3}{2} \min(A_i, B_i) \\
\max(A_j, B_j) &\leq \frac{3}{2} \min(A_j, B_j) \\
\min(A_i, B_i) \leq \ell(A_i, B_i) &\leq \frac{5}{4} \min(A_i, B_i) \\
\min(A_j, B_j) \leq \ell(A_j, B_j) &\leq \frac{5}{4} \min(A_j, B_j) \\
\end{align*}

The algorithm ensures the following:
\begin{align*}
\min(A_j, B_j) &\leq \min(Q) \leq \max(A_i, B_i)\\
\min(A_i, B_i) &\leq \min(Q) \leq \max(A_j, B_j)\\
\end{align*}

Combining these we have:
\begin{align*}
\ell(A_i, B_i) \leq \frac{5}{4} \min(A_i, B_i) &\leq \frac{5}{4} \frac{3}{2} \min(A_j, B_j) < 2 \ell(A_j, B_j)\\
\ell(A_i, B_i) &\leq 2 \ell(A_j, B_j) \\
\ell(A_j, B_j) \leq \frac{5}{4} \min(A_j, B_j) &\leq \frac{5}{4} \frac{3}{2} \min(A_i, B_i) < 2 \ell(A_i, B_i)\\
\ell(A_j, B_j) / 2 &\leq \ell(A_i, B_i) \\
\end{align*}

It follows that $\ell(A_j, B_j) \in [\ell(A_i, B_i)/2, 2\ell(A_i, B_i)]$.
\end{proof}

\begin{theorem} \label{theorem:full}
Algorithm \qproc{GreedySpanner} computes the greedy spanner for dilation $t$ in $O\left(n^2 \log^2 n \frac{1}{(t-1)^{3d}}+ n^2 \log n \frac{1}{(t-1)^{4d}}\right)$ time and $O\left(\frac{1}{(t-1)^d} n\right)$ space.
\end{theorem}

\begin{proof}
We can easily precompute which well-separated pairs are close to a particular well-separated pair as needed in line \ref{algo:loop} in $O(m^2)$ time, without affecting the running time. By fact \ref{fact:packing} there are only at most $c_{s t}$ such well-separated pairs per well-separated pair, so this uses $O\left(\frac{1}{(t-1)^d} n\right)$ space.

Combining Observation \ref{observation:greedyedges} with Lemma \ref{lemma:dijkstra} we conclude that we can compute \qproc[i]{ClosestPair} in
\begin{align*}
O\left(\min(|A_i|, |B_i|)\left(n \log n + \frac{1}{(t-1)^d} n\right)\right)
\end{align*}
time and $O(n)$ space. By Lemma \ref{lemma:closestpaircount} the time taken by all \qproc[i]{ClosestPair} calls is therefore
\begin{align*}
O\left(\sum_{i=1}^m (1 + c_{s t}) \min(|A_i|, |B_i|)\left(n \log n + \frac{1}{(t-1)^d} n\right)\right)
\end{align*}
and its space usage is $O(n)$ by reusing the space for the calls. Using Fact \ref{fact:manyone}, this is at most
\begin{align*}
O\left(\frac{1}{(t-1)^d}\left(1+\frac{t}{t-1}\right)^d\frac{1}{(t-1)^d} n \log n \left(n \log n + \frac{1}{(t-1)^d} n\right) \right)
\end{align*}
which simplifies to
\begin{align*}
O\left(n^2 \log^2 n \frac{1}{(t-1)^{3d}} + n^2 \log n \frac{1}{(t-1)^{4d}}\right)\;.
\end{align*}

The time taken by all other steps of the algorithm is insignificant compared to the time used by \qproc[i]{ClosestPair} calls. These other steps are: computing the WSPD and all actions with regard to the queue. All these other actions use $O\left(\frac{1}{(t-1)^d} n\right)$ space. Combining this with Theorem \ref{theorem:correctness}, the theorem follows.
\end{proof}

\section{Making the algorithm practical} \label{section:practical}

Experiments suggested that implementing the above algorithm as-is does not yield a practical algorithm. With the four optimizations described in the following sections, the algorithm attains running times that are a small constant slower than the algorithm introduced in \cite{FarshiG09} called FG-greedy, which is considered the best practical algorithm known in literature.

\subsection{Finding close-by pairs}
The algorithm at some point needs to know which pairs are `close' to the pair for which we are currently adding an edge. In our proof above, we suggested that these pairs be precomputed in $O(m^2)$ time. Unfortunately, this precomputation step turns out to take much longer than the rest of the algorithm. If $n=100$, then (on a uniform pointset) $m \approx 2000$ and $m^2 \approx 4000000$ while the number of edges $e$ in the greedy spanner is about 135. Our solution is to simply find them using a linear search every time we need to know this information. This only takes $O(e \cdot m)$ time, which is significantly faster.

\subsection{Reducing the number of Dijkstra computations}

After decreasing the time taken by preprocessing, the next part that takes the most time are the Dijkstra computations, whose running time dwarfs the rest of the operations. We would therefore like to optimize this part of the algorithm. For every well-separated pair, we save the length of the shortest path found by any Dijkstra computation performed on it, that is, its source $s$, target $t$ and distance $\delta(s, t)$. Then, if we are about to perform a Dijkstra computation on a vertex $u$, we first check if the saved path is already good enough to `cover' all nodes in $B_i$. Let $c$ be the center of the circle around the bounding box of $B_i$ and $r$ its radius. We check if $t \cdot |us| + \delta(s, t) + t \cdot (|tc| + r) \leq t \cdot (|uc| - r)$ and mark it as `irrelevant for the rest of the algorithm'. This optimization roughly improves its running time by a factor three.

\subsection{Sharpening the bound of Observation \ref{observation:faraway}}

The bound given in Observation \ref{observation:faraway} can be improved. Let $\{ A_i, B_i \}$ be the well-separated pair for which we just added an edge and let $\{ A_j, B_j \}$ be the well-separated pair under consideration in our linear search. First, some notation: let $X_k, X_l$ be sets belonging to some well-separated pair (not necessarily the same pair), then $\min(X_k, X_l)$ denotes the (shortest) distance between the two circles around the bounding boxes of $X_k$ and $X_l$ and $\max(X_k, X_l)$ the longest distance between these two circles. Let $\ell = \ell(A_i, B_i)$. We can then replace the condition of Lemma \ref{observation:faraway} by the sharper condition $\min(A_i, A_j) + \ell + \min(B_j, B_i) \leq t \cdot \max(A_j, B_j) \vee \min(A_i, B_j) + \ell + \min(A_j, B_i) \leq t \cdot \max(B_j, A_j)$ The converse of the condition implies that the edge just added cannot be part of a $t$-path between a node in $\{ A_j, B_j \}$, so the correctness of the algorithm is maintained. This leads to quite a speed increase.

\subsection{Miscellaneous optimizations}

There are two further small optimizations we have added to our implementation.

Firstly, rather than using the implicit linear space representation of the WSPD, we use the explicit representation where every node in the split tree stores the points associated with that node. For point sets where the ratio of the longest and the shortest distance is bounded by some polynomial in $n$, this uses $O(n \log n)$ space rather than $O(n)$ space. This is true for all practical cases, which is why we used it in our implementation. For arbitrary point sets, this representation uses $O(n^2)$ space. In practice, this extra space usage is hardly noticeable and it speeds up access to the points significantly.

Secondly, rather than performing Dijkstra's algorithm, we use the $A^*$ algorithm. This algorithm uses geometric estimates to the target to guide the computation to its goal, thus reducing the search space of the algorithm \cite{Goldberg:2005:CSP:1070432.1070455}.

We have tried a number of additional optimizations, but none of them resulted in a speed increase. We describe them here.

We have tried to replace $A^*$ by $ALT$, a shortest path algorithm that uses landmarks -- see \cite{Goldberg:2005:CSP:1070432.1070455} for details on ALT -- which gives better lower bounds than the geometric estimates used in $A^*$. However, this did not speed up the computations at all, while costing some amount of overhead.

We have also tried to further cut down on the number of Dijkstra computations. We again used that we store the lengths of the shortest paths found so far per well-separated pair. Every time after calling \qproc[i]{ClosestPair} we checked if the newly found path is `good enough' for other well-separated pairs, that is, if the path combined with $t$-paths from the endpoints of the well-separated pairs would give $t$-paths for all pairs of points in the other well-separated pair. This decreased the number of Dijkstra computations performed considerably, but the overhead from doing this for all pairs was greater than its gain.

We tried to speed up the finding of close-by pairs by employing range trees. We also tried using the same range trees to perform the optimization of the previous paragraph only to well-separated pairs `close by' our current well-separated pair. Both optimizations turned out to give a speed increase and in particular the second retained most of its effectiveness even though we only tried it on close-by pairs, but the overhead of range trees was vastly greater than the gain -- in particular the space usage of range trees made the algorithm use about as much space as the original greedy algorithms.

\section{Experimental results} \label{section:results}

We have run our algorithm on point sets whose size ranged from 100 to 1,000,000 vertices on several distributions. If the set contained at most 10,000 points, we have also run the FG-greedy algorithm to compare the two algorithms. We have recorded both space usage and running time (wall clock time). We have also performed a number of tests with decreasing values of $t$ on datasets of size 10,000 and 50,000. Finally, as this is the first time we can compute the greedy spanner on large graphs, we have compared it to the $\Theta$-graph and WSPD-based spanners on large instances.

We have used three kinds of distributions from which we draw our points: a uniform distribution, a gamma distribution with shape parameter 0.75, and a distribution consisting of $\sqrt{n}$ uniformly distributed pointsets of $\sqrt{n}$ uniformly distributed points. The results from the gamma distribution were nearly identical to those of the uniform pointset, so we did not include them. All our pointsets are two-dimensional.

\subsection{Experiment environments}

The algorithms have been implemented in C++. We have implemented all data structures not already in the \verb|std|. The random generator used was the Mersenne Twister PRNG -- we have used a C++ port by J. Bedaux of the C code by the designers of the algorithm,  M. Matsumoto and T. Nishimura.

We have used two servers for the experiments. Most experiments have been run on the first server, which uses an Intel Core i5-3470 (3.20GHz) and 4GB (1600 MHz) RAM. It runs the Debian 6.0.7 OS and we compiled for 32 bits using G++ 4.7.2 with the -O3 option. For some tests we needed more memory, so we have used a second server. This server uses an Intel Core i7-3770k (3.50GHz) and 32 GB RAM. It runs Windows 8 Enterprise and we have compiled for 64 bits using the Microsoft C++ compiler (17.00.51106.1) with optimizations turned on.

\subsection{Dependence on instance size}

Our first set of tests compared FG-greedy and our algorithm for different values of $n$. The results are plotted in Fig.~\ref{figure:plot}. As FG-greedy could only be ran on relatively small instances, its data points are difficult to see in the graph, so we added a zoomed-in plot for the bottom-left part of the plot.

\begin{figure}[h!]\centering
\includegraphics[width=6cm]{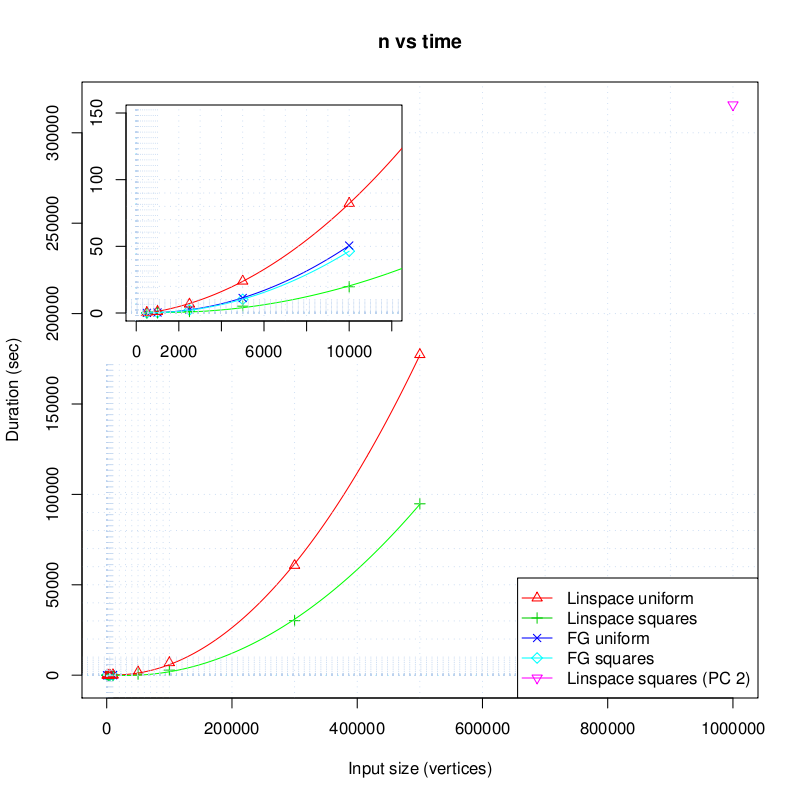}
\includegraphics[width=6cm]{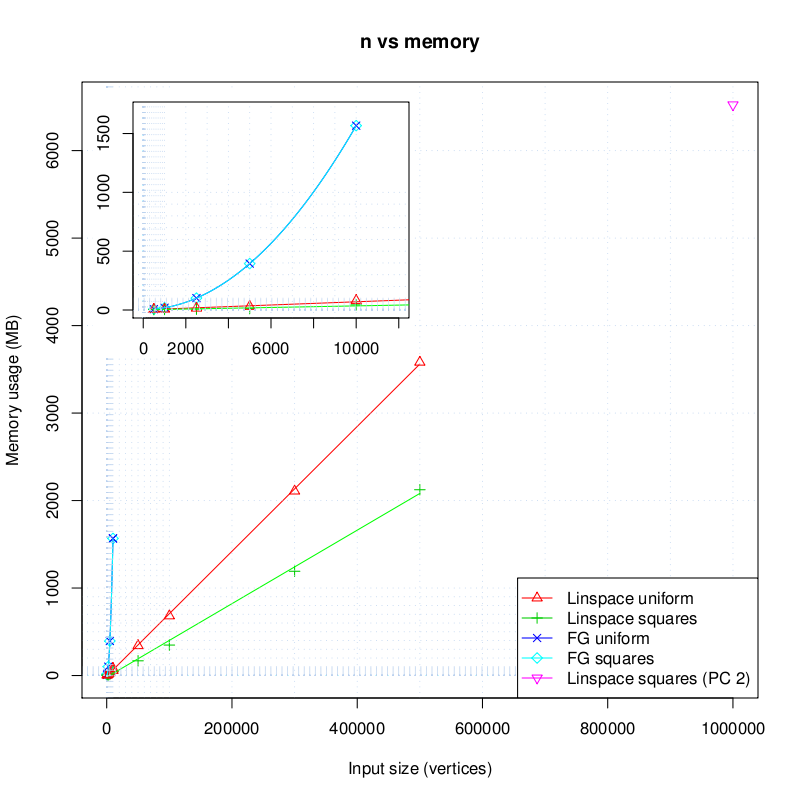}
\caption{The left plot shows the running time of our algorithm on uniform and clustered data for variously sized instances. The right plot shows the memory usage of our algorithm on the same data. The lines are fitted quadratic (right) and linear (left) curves. The outlier at the right side was from an experiment performed on a different server. Results for FG-greedy are also shown but were near-impossible to see, so a zoomed-in view of the leftmost corner of both plots is included in the top-left of both plots. The memory usage explosion of FG-greedy is visible in the right plot.}
\label{figure:plot}
\end{figure}

We have used standard fitting methods to our data points: the running time of all algorithms involved fits a quadratic curve well, the memory usage of our algorithm is linear and the memory usage of FG-greedy is quadratic. This nicely fits our theoretical analysis. In fact, the constant factors seem to be much smaller than the bound we gave in our proof. We do note a lack of `bumps' that are often occur when instance sizes start exceeding caches: this is probably due to the cache-unfriendly behavior of our algorithm and the still significant constant factor in our memory usage that will fill up caches quite quickly.

Comparing our algorithm to FG-greedy, it is clear that the memory usage of our algorithm is vastly superior. The plot puts into perspective just how much larger the instances are that we are able to deal with using our new algorithm compared to the old algorithms. Furthermore, our algorithm is about twice as fast as FG-greedy on the clustered datasets, and only about twice as slow on uniform datasets. On clustered datasets the number of computed well-separated pairs is much smaller than on uniform datasets so this difference does not surprise us. These plots suggest that our aim -- roughly equal running times at vastly reduced space usage -- is reached with this algorithm.

\subsection{Dependence on $t$}

We have tested our algorithms on datasets of 10,000 and 50,000 points, setting $t$ to $1.1$, $1.2$, $1.4$, $1.6$, $1.8$ and $2.0$ to test the effect of this parameter. The effects of the parameter ended up being rather different between the uniform and clustered datasets.

\begin{figure}[h!]\centering
\includegraphics[width=6cm]{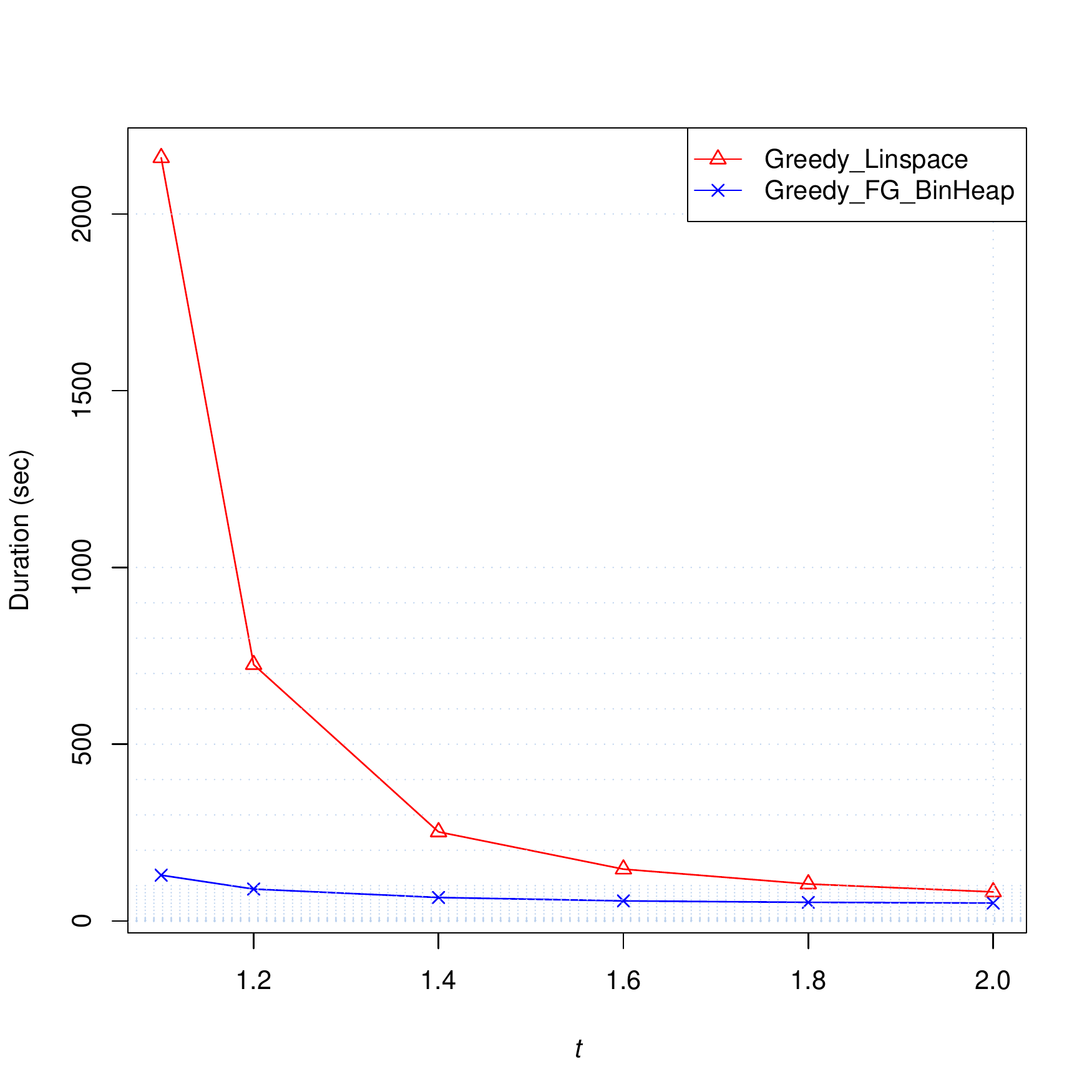}
\includegraphics[width=6cm]{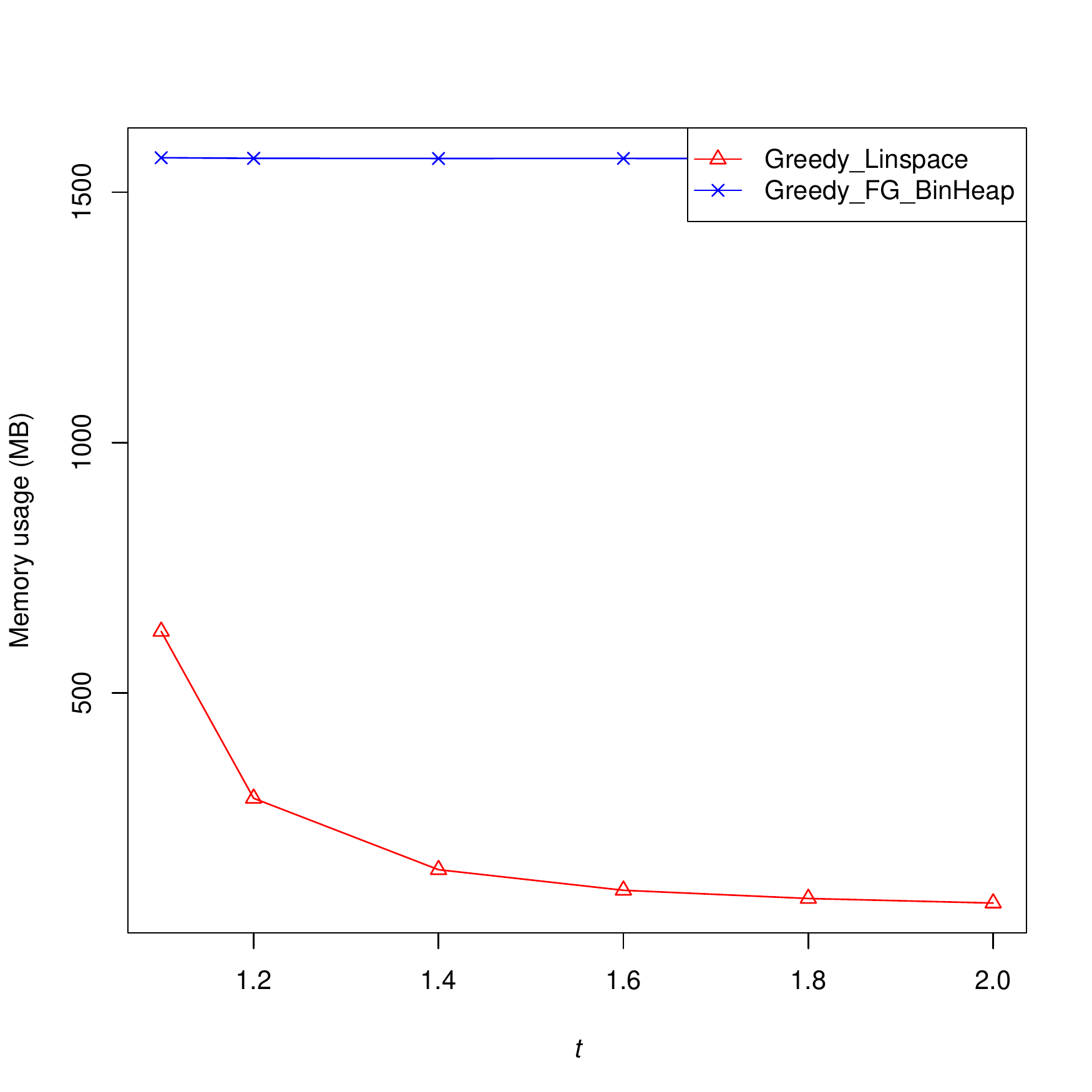}
\caption{The left plot shows the running time of our algorithm on a dataset of 10,000 uniformly distributed points for various values of $t$. The right plot shows the memory usage of our algorithm for the same settings.}
\label{figure:uniplot1}
\end{figure}

\begin{figure}[h!]\centering
\includegraphics[width=6cm]{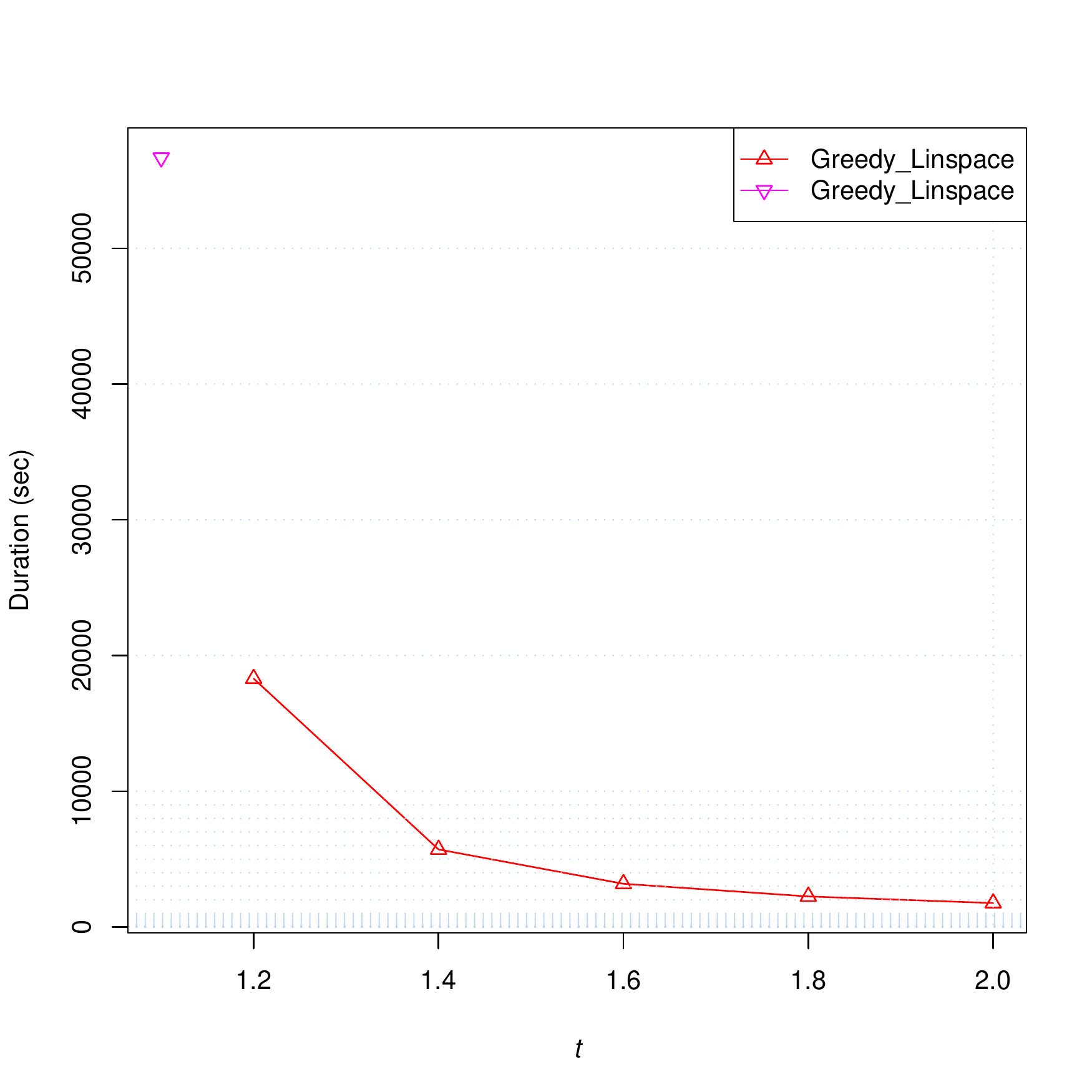}
\includegraphics[width=6cm]{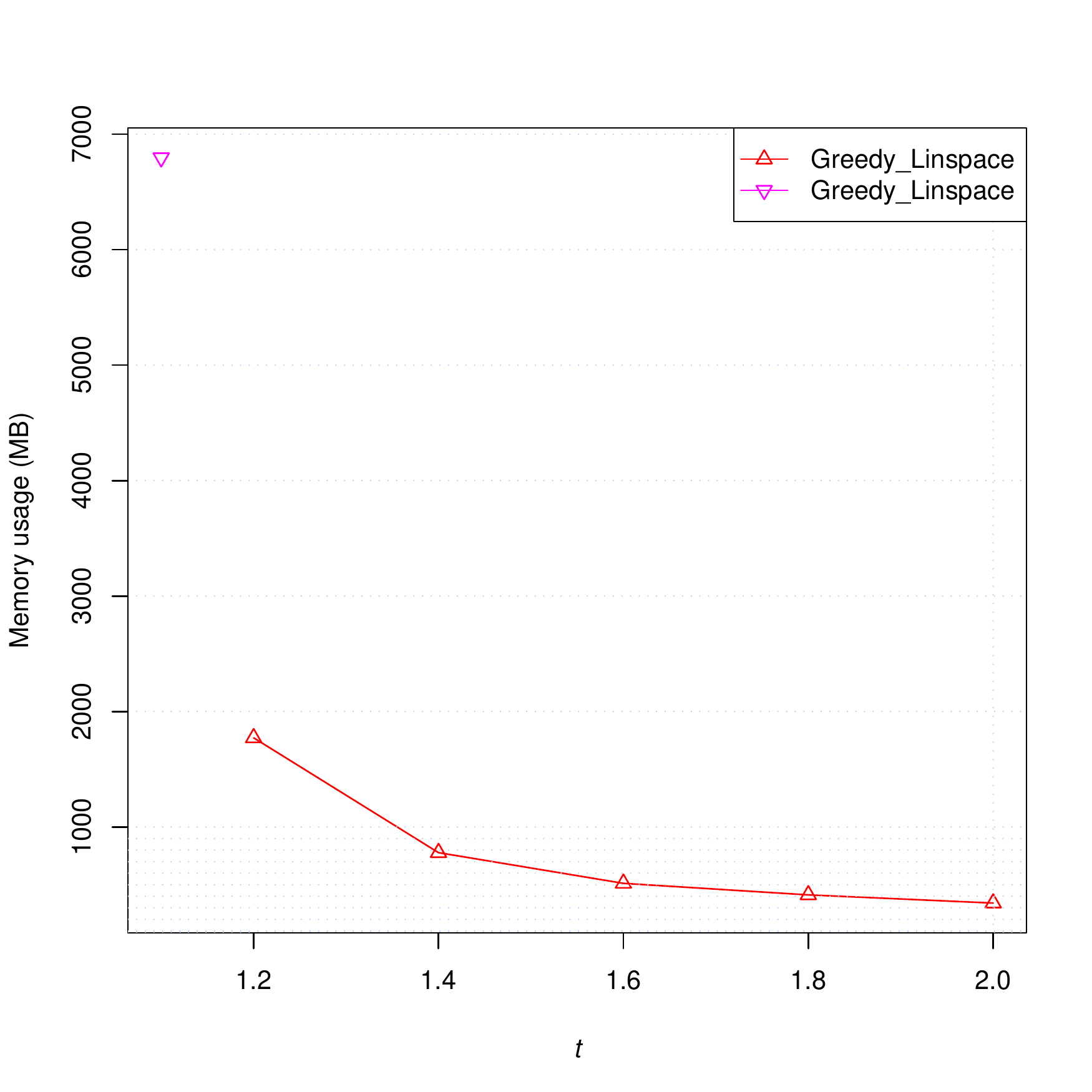}
\caption{The left plot shows the running time of our algorithm on a dataset of 50,000 uniformly distributed points for various values of $t$. The right plot shows the memory usage of our algorithm for the same settings.}
\label{figure:uniplot2}
\end{figure}

On uniform pointsets, see Figures \ref{figure:uniplot1} and \ref{figure:uniplot2}, our algorithm is about as fast as FG-greedy when $t=2$, but its performance degrades quite rapidly as $t$ decreases compared to FG-greedy. A hint to this behavior is given by the memory usage of our algorithm: it starts vastly better but as $t$ decreases it becomes only twice as good as FG-greedy. This suggests that the number of well-separated pairs grows rapidly as $t$ decreases, which explains the running time decrease.

\begin{figure}[h!]\centering
\includegraphics[width=6cm]{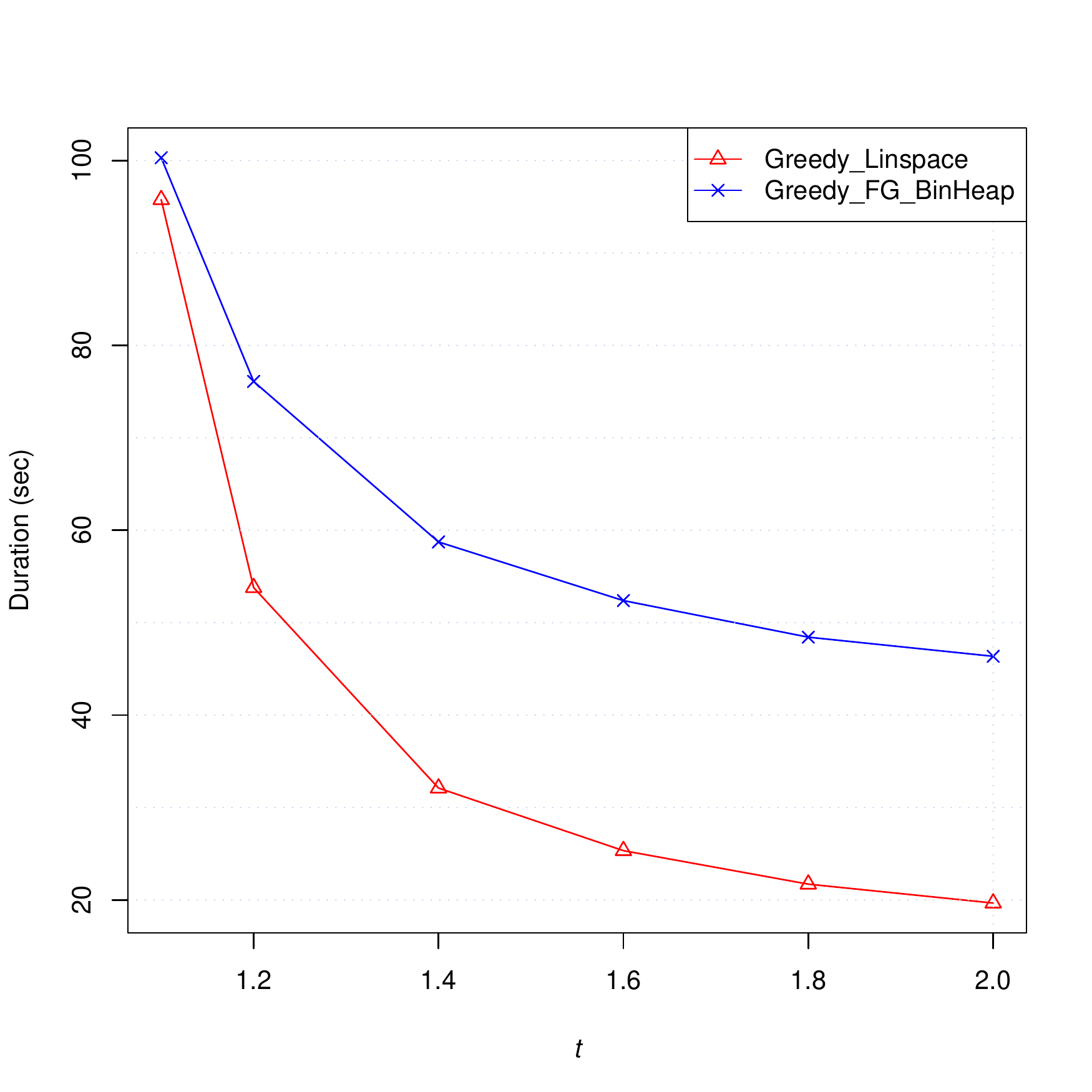}
\includegraphics[width=6cm]{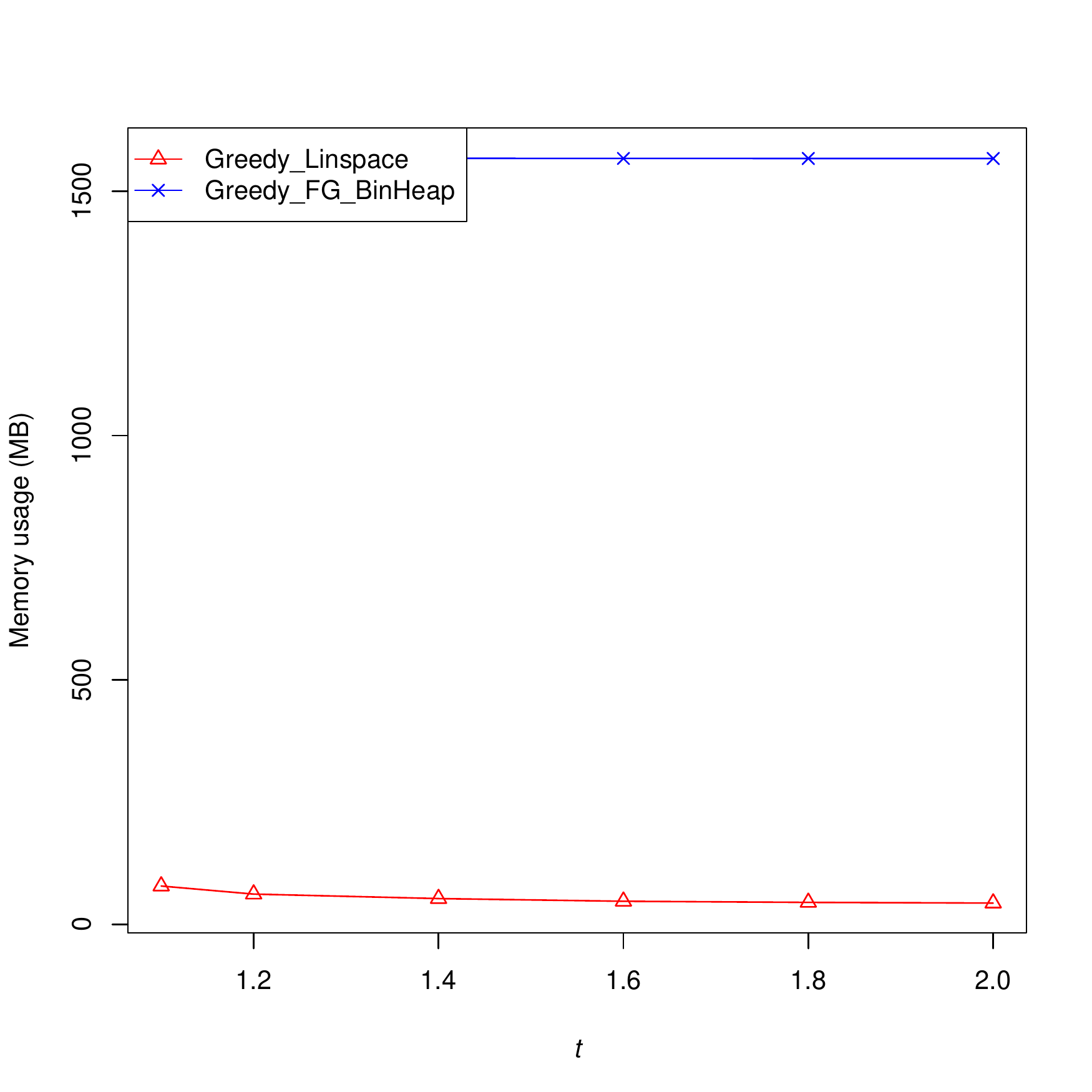}
\caption{The left plot shows the running time of our algorithm on a dataset of 10,000 clustered points for various values of $t$. The right plot shows the memory usage of our algorithm for the same settings.}
\label{figure:clusteredplot1}
\end{figure}

\begin{figure}[h!]\centering
\includegraphics[width=6cm]{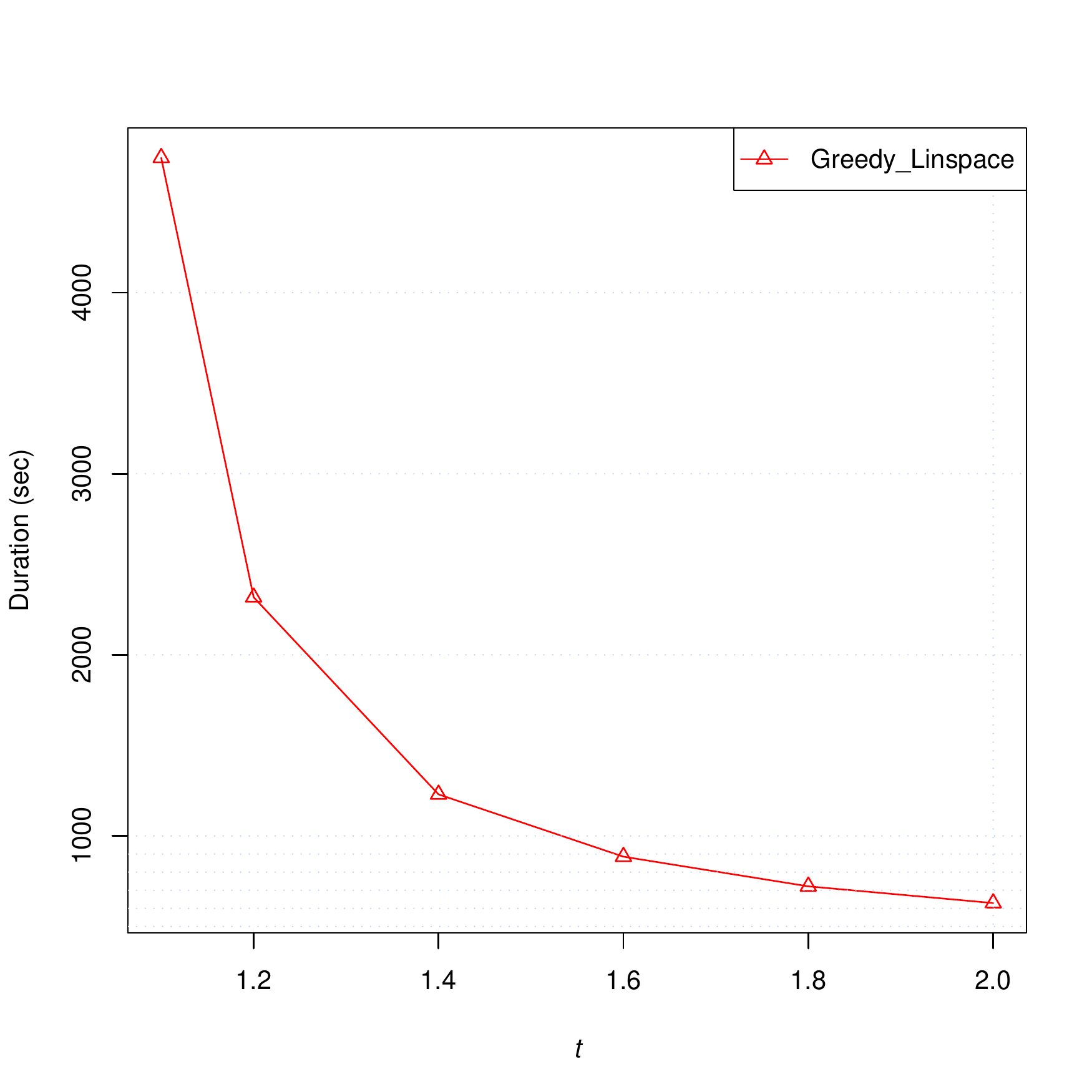}
\includegraphics[width=6cm]{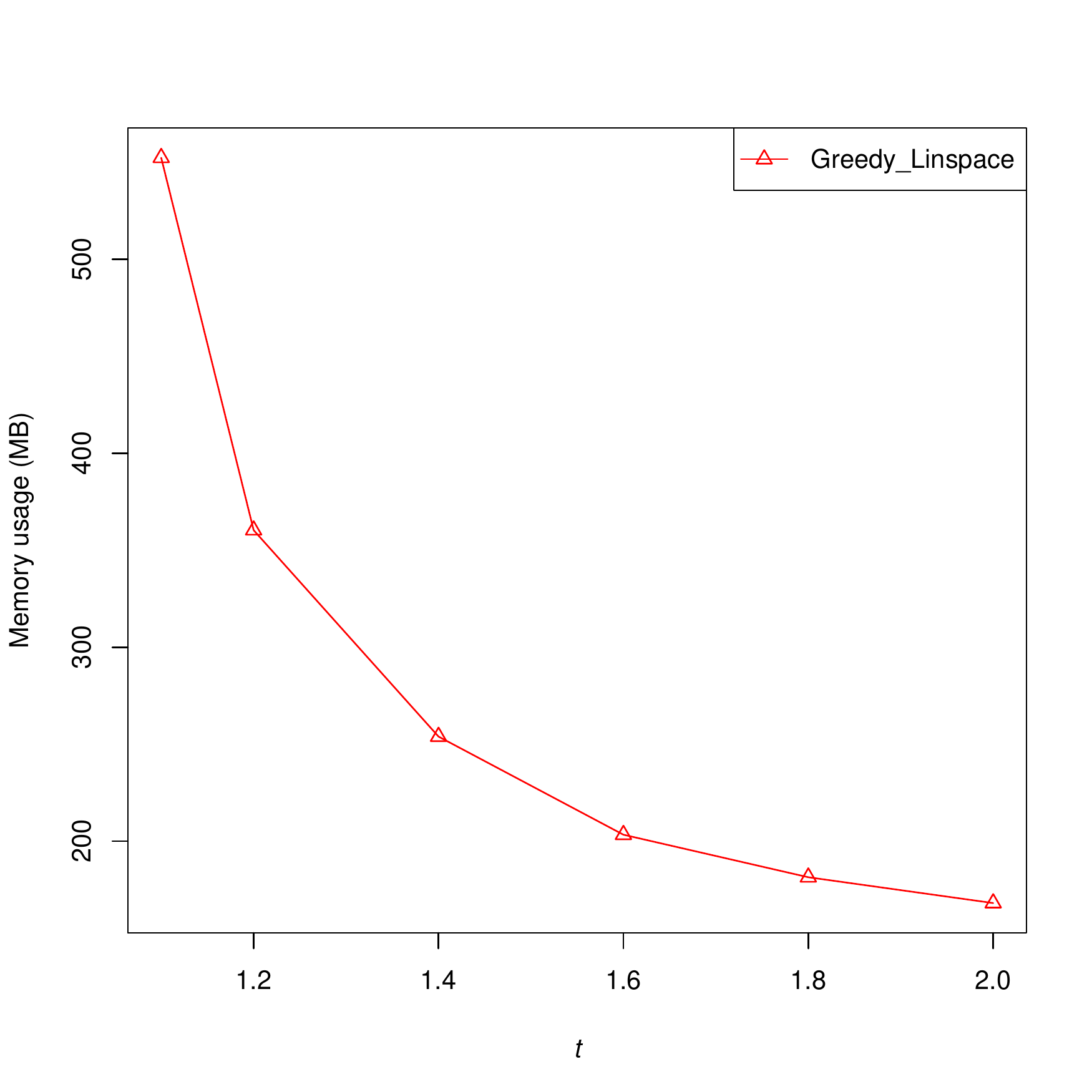}
\caption{The left plot shows the running time of our algorithm on a dataset of 50,000 clustered points for various values of $t$. The right plot shows the memory usage of our algorithm for the same settings.}
\label{figure:clusteredplot2}
\end{figure}

On clustered pointsets, see Figures \ref{figure:clusteredplot1} and \ref{figure:clusteredplot2}, the algorithms compare very differently. FG-greedy starts out twice as slow as our algorithm when $t=2$ and when $t=1.1$, our algorithm is only slightly faster than FG-greedy. The memory usage of our algorithm is much less dramatic than in the uniform point case: it hardly grows with $t$ and therefore stays much smaller than FG-greedy. The memory usage of FG-greedy only depends on the number of points and not on $t$ or the distribution of the points, so its memory usage is the same.

\subsection{Comparison with other spanners}

We have computed the greedy spanner on the instance shown in Fig.~\ref{figure:greedy}, which has 115,475 points. On this instance the greedy spanner for $t=2$ has 171,456 edges, a maximum degree of 5 and a weight of 11,086,417. On the same instance, the $\Theta$-graph with $k=6$ has 465,230 edges, a maximum degree of 62 and a weight of 53,341,205. The WSPD-based spanner has 16,636,489 edges, a maximum degree of 1,271 and a weight of 20,330,194,426.

As shown in Fig.~\ref{figure:plot}, we have computed the greedy spanner on a pointset of 500,000 uniformly distributed points. On this instance the greedy spanner for $t=2$ has 720,850 edges, a maximum degree of 6 and a weight of 9,104,690. On the same instance, the $\Theta$-graph with $k=6$ has 2,063,164 edges, a maximum degree of 22 and a weight of 39,153,380. We were unable to run the WSPD-based spanner algorithm on this pointset due to its memory usage.

As shown in Fig.~\ref{figure:plot}, we have computed the greedy spanner on a pointset of 1,000,000 clustered points. On this instance the greedy spanner for $t=2$ has 1,409,946 edges, a maximum degree of 6 and a weight of 4,236,016. On the same instance, the $\Theta$-graph with $k=6$ has 4,157,016 edges, a maximum degree of 135 and a weight of 59,643,264. We were unable to run the WSPD-based spanner algorithm on this pointset due to its memory usage.

We have computed the greedy spanner on a pointset of 50,000 uniformly distributed points with $t=1.1$. On this instance the greedy spanner has 225,705 edges, a maximum degree of 18 and a weight of 15,862,195. On the same instance, the $\Theta$-graph with $k=73$ (which is the smallest $k$ for which a guarantee of $t=1.1$ has been proven to our knowledge) has 2,396,361 edges, a maximum degree of 146 and a weight of 495,332,746. We were unable to run the WSPD-based spanner algorithm on this pointset with $t=1.1$ due to its memory usage.

These results show that the greedy spanner really is an excellent spanner, even on large instances and for low $t$, as predicted by its theoretical properties.

\section{Conclusion}

We have presented an algorithm that computes the greedy spanner in Euclidean space in $O(n^2 \log^2 n)$ time and $O(n)$ space for any fixed stretch factor and dimension. Our algorithm avoids computing all distances by considering well-separated pairs instead. It consists of a framework that computes the greedy spanner given a subroutine for a bichromatic closest pair problem. We give such a subroutine which leads to the desired result.

We have presented several optimizations to the algorithm. Our experimental results show that these optimizations make our algorithm have a running time close to the fastest known algorithms for the greedy spanner, while massively decreasing space usage. It allowed us to compute the greedy spanner on very large  instances of a million points, compared to the earlier instances of at most 13,000 points. Given that our algorithm is the first algorithm with a near-quadratic running time guarantee that has actually been implemented, that it has linear space usage and that its running time is comparable to the best known algorithms, we think our algorithm is the method of choice to compute greedy spanners.

We leave open the problem of providing a faster subroutine for solving the \emph{bichromatic closest pair with dilation larger than $t$} problem in our framework, which may allow the greedy spanner to be computed in subquadratic time. Particularly the case of the Euclidean plane seems interesting, as the closely related `ordinary' \emph{bichromatic closest pair} problem can be solved quickly in this setting.

\bibliography{refs}

\newcommand{\SortNoop}[1]{}\def\doi#1{$\!$\marginpar
  [\hspace*{3.3cm}\href{http://dx.doi.org/#1}{\tiny \color{red} doi}]
  {\hspace*{-12.9cm}\href{http://dx.doi.org/#1}{\tiny \color{red}
  doi}}}\def\citeurl#1{$\!$\marginpar [\hspace*{3.3cm}\href{#1}{\tiny
  \color{red} url}] {\hspace*{-12.9cm}\href{#1}{\tiny \color{red}
  url}}}\def\ignore#1{}
\begin{thebibliography}{10}

\bibitem{BoseCFMS2010}
P.~Bose, P.~Carmi, M.~Farshi, A.~Maheshwari, and M.~Smid.
\newblock Computing the greedy spanner in near-quadratic time.
\newblock {\em Algorithmica}, 58(3):711--729, 2010.

\bibitem{Callahan95dealingwith}
P.~B. Callahan.
\newblock {\em Dealing with Higher Dimensions: The Well-Separated Pair
  Decomposition and Its Applications}.
\newblock PhD thesis, Johns Hopkins University, Baltimore, Maryland, 1995.

\bibitem{Callahan:1995:DMP:200836.200853}
P.~B. Callahan and S.~R. Kosaraju.
\newblock A decomposition of multidimensional point sets with applications to
  k-nearest-neighbors and n-body potential fields.
\newblock {\em J. ACM}, 42(1):67--90, 1995.

\bibitem{Chew1989}
L.~P. Chew.
\newblock There are planar graphs almost as good as the complete graph.
\newblock {\em J. Comput. System Sci.}, 39(2):205 -- 219, 1989.

\bibitem{FarshiG09}
M.~Farshi and J.~Gudmundsson.
\newblock Experimental study of geometric {\it t}-spanners.
\newblock {\em ACM J. Experimental Algorithmics}, 14, 2009.

\bibitem{GaoGHZZ05}
J.~Gao, L.~J. Guibas, J.~Hershberger, L.~Zhang, and A.~Zhu.
\newblock Geometric spanners for routing in mobile networks.
\newblock {\em IEEE J. Selected Areas in Communications}, 23(1):174--185, 2005.

\bibitem{Goldberg:2005:CSP:1070432.1070455}
A.~V. Goldberg and C.~Harrelson.
\newblock Computing the shortest path: A search meets graph theory.
\newblock In {\em 16th ACM-SIAM Sympos. Discrete Algorithms}, pages 156--165.
  SIAM, 2005.

\bibitem{DilationandDetours}
J.~Gudmundsson and C.~Knauer.
\newblock Dilation and detours in geometric networks.
\newblock In T.~Gonzales, editor, {\em Handbook on Approximation Algorithms and
  Metaheuristics}, pages 52--1 -- 52--16. Chapman \& Hall/CRC, Boca Raton,
  2006.

\bibitem{Keil:1988:ACE:61764.61787}
J.~M. Keil.
\newblock Approximating the complete euclidean graph.
\newblock In {\em 1st Scandinavian Workshop on Algorithm Theory (SWAT)}, volume
  318 of {\em LNCS}, pages 208--213. Springer, 1988.

\bibitem{Narasimhan:2007:GSN:1208237}
G.~Narasimhan and M.~Smid.
\newblock {\em Geometric Spanner Networks}.
\newblock Cambridge University Press, New York, NY, USA, 2007.

\bibitem{JGT:JGT3190130114}
D.~Peleg and A.~A. Sch\"affer.
\newblock Graph spanners.
\newblock {\em Journal of Graph Theory}, 13(1):99--116, 1989.

\end{thebibliography}
\bibliographystyle{abbrv}

\end{document}